\documentclass{sigplanconf}
\pdfoutput=1

\makeatletter                   
\def\mdseries@tt{m}             
\makeatother                    
\usepackage{tikz}

\usepackage[draft=true]{minted} 
\usepackage{color}
\usepackage[pdftex]{hyperref}           
\hypersetup{
    colorlinks=true,
    linkcolor=blue,
    filecolor=red,      
    urlcolor=magenta,
    breaklinks=true,            
}
\usepackage{breakurl}           

\newcommand\inapx[1]{in Appendix #1}

\title{Reducing Commutativity Verification to Reachability with Differencing Abstractions}

\authorinfo{Eric Koskinen}{Stevens Institute of Technology}{}
\authorinfo{Kshitij Bansal}{Google}{}

\newcommand\CityProver{{\sc CityProver}}
\newcommand\todo[1]{}

\newcommand\red[1]{{\color{red} #1}}

\newcommand\ignore[1]{}

\renewcommand\red[1]{\textcolor{red}{#1}}

\usepackage{amsmath,amssymb,amsthm}
\usepackage{proof}

\newtheorem{definition}{Definition}[section]
\newtheorem{theorem}{Theorem}[section]

\usepackage{graphicx}
\usepackage{courier}
\usepackage{mathtools}
\usepackage{extpfeil}
\usepackage{enumerate}
\usepackage{color}
\usepackage{program}
\usepackage{proof}

\usepackage{textcomp}
\usepackage[scaled]{beramono}
\usepackage{listings}
\lstdefinestyle{MyC}{
language=C++,
tabsize=2,
numbers=none,
captionpos=t,
basicstyle=\sffamily,
commentstyle=\sffamily,
breaklines=true,
columns=fullflexible,
mathescape=true,
escapechar=`}

\newcommand\ttt[1]{\text{\lstinline{#1}}}

\newcommand\sem[1]{[\![ #1 ]\!]}

\newcommand\opeq{\simeq}

\newcommand\Xs{\bar{x}}
\newcommand\Ys{\bar{y}}
\newcommand\Rms{\bar{r}}
\newcommand\Rns{\bar{s}}
\newcommand\As{\bar{a}}
\newcommand\Bs{\bar{b}}
\newcommand\Us{\bar{u}}
\newcommand\Vs{\bar{v}}

\newcommand\cfarightarrow[1]{\xtwoheadrightarrow{#1}}

\newcommand\Ecfarightarrow[1]{\xtwoheadrightarrow[\enct]{#1}}

\newcommand\vln{\theta}
\newcommand\Vlns{\Theta}

\newcommand\Ops{\red{Ops}}

\newcommand\aut{\mathcal{A}}

\newcommand\myM{\addx{}}
\newcommand\myN{\isiny{}}
\newcommand\rvm{\texttt{r}_{\ttt{add}}}
\newcommand\rvn{\texttt{r}_{\ttt{isin}}}
\newcommand\rvmn{\texttt{r}_{\ttt{add;isin}}}
\newcommand\rvnm{\texttt{r}_{\ttt{isin;add}}}

\newenvironment{itemize*}%
  {\begin{itemize}%
}
  {\end{itemize}}
\newenvironment{enumerate*}%
  {\begin{enumerate}%
}
  {\end{enumerate}}

\definecolor{light-gray}{gray}{0.85}

\newcommand\mndifferencing{$mn$-differencing}

\newcommand\figboxN[1]{\noindent{\ \\\fbox{\begin{minipage}{3.2in}
#1
\end{minipage}
}}}

\renewcommand\topfraction{.9}

\renewcommand\floatpagefraction{.8}

\newcommand\ie{{\it i.e.}}
\newcommand\eg{{\it e.g.}}

\newcommand\FULL[1]{\red{(omitted content)}}

\newcommand\filled{\ttt{filled}}

\newcommand\addx{\ttt{add}(\ensuremath{x})}
\newcommand\pushx{\ttt{push}(\ensuremath{x})}
\newcommand\isinx{\ttt{isin}(\ensuremath{x})}
\newcommand\isiny{\ttt{isin}(\ensuremath{y})}
\renewcommand\aa{\ttt{a}}

\newcommand\bb{\ttt{b}}

\newcommand\PCai{\varphi_{{\addx}}^{{\isiny}}}
\newcommand\topp{\ttt{top}}

\newcommand\PCpp{\varphi_{{\pushx}}^{{\textsf{pop}}}}

\newcommand\autencMmn[1]{\mathcal{A}(#1)}
\newcommand\commuai{\varphi_{\addx}^{\isiny}}
\newcommand\autencMai{\autencMmn{\commuai}}
\newcommand\autencMN{\autencMmn{\varphi_m^n}}

\renewcommand{\topfraction}{.95}

\begin{document}
\toappear{}
\maketitle

\begin{abstract}
Commutativity of data structure methods is of ongoing interest, with roots in the database community. In recent years commutativity has been shown to be a key ingredient to enabling multicore concurrency in contexts such as parallelizing compilers, transactional memory,  speculative execution and, more broadly, software scalability. Despite this interest, it remains an open question as to how a data structure's commutativity specification can be verified automatically from its implementation.

In this paper, we describe techniques to automatically prove the correctness of method commutativity conditions from data structure implementations. We introduce a new kind of abstraction that characterizes the ways in which the effects of two methods differ depending on the order in which the methods are applied, and abstracts away effects of methods that would be the same regardless of the order. We then describe a novel algorithm that reduces the problem to reachability, so that off-the-shelf program analysis tools can perform the reasoning necessary for proving commutativity. Finally, we describe a proof-of-concept implementation  and experimental results, showing that   our tool can verify commutativity of data structures such as a memory cell, counter,  two-place Set, array-based stack, queue, and a rudimentary hash table. We conclude with a discussion of what makes a data structure's commutativity provable with today's tools and what needs to be done to prove more in the future.
\todo{downplay the reduction to reachability - instance of 2-safety}
\todo{remove $I$ in experiments?}
\todo{mkSupp. add appendix.pdf}
\todo{bitset}

\end{abstract}

\section{Introduction}
\label{sec:intro}

For an object $o$, with state $\sigma$ and methods $m$, $n$, etc.,
let $\Xs$ and $\Ys$ denote argument vectors and $o.m(\Xs)/\Rms$ denote 
a method signature, including a vector of corresponding return values $\Rms$.
\emph{Commutativity} of two methods, denoted $o.m(\Xs)/\Rms \bowtie o.n(\Ys)/\Rns$, 
are
circumstances where 
operations $m$ and $n$, when applied in either order, lead to the same
final state and agree on the intermediate return values $\Rms$ and
$\Rns$.
A \emph{commutativity condition} is a logical formula
$\varphi_m^n(\sigma,\Xs,\Rms,\Ys,\Rns)$ indicating, for a given state $\sigma$, whether the two operations will always commute, as a function of parameters and return values.

Commutativity conditions are typically much smaller than full specifications,
yet they are powerful: it has been shown that they are an enabling
ingredient in correct, efficient concurrent execution in the context of parallelizing compilers~\cite{rinard}, transactional memory~\cite{ppopp08,Ni2007,Koskinen2015}, optimistic parallelism~\cite{Kulkarni2007}, speculative execution, features~\cite{Chechik2018}, etc.  More broadly, a recent paper from the systems community~\cite{DBLP:journals/tocs/ClementsKZMK15} found that,  when software fragments are implemented so that they commute, better scalability is achieved.
  Intuitively, commutativity captures independence and, if two code fragments commute then, when combined with linearizability (for which proof techniques exist, \eg,\cite{vafeiadis2010,Bouajjani2015}) they can be executed concurrently.
%
To employ this approach to concurrency
it is important that commutativity be \emph{correct}
and, in recent years, growing effort has been made toward reasoning about  commutativity conditions automatically. At present, these works are either unsound~\cite{DBLP:conf/cav/GehrDV15,aleen} or else 
they rely on data structure specifications as  intermediaries~\cite{KR:PLDI11,tacas18} (See Section~\ref{sec:relwork}).

\todo{Intuitively, commutativity  is a multi-trace property: relating the behaviors in one circumstance with those in another. It is therefore tempting to pose the problem as a $k$-safety problem, along the lines of the following:
\[\begin{array}{c}
    \{ \sigma_1 = \sigma_2 \wedge \varphi_m^n(\sigma,\As,\Bs) \}\\
    \begin{array}{l|l}
       r_m^1 := m(\As); & r_n^2 := n(\Bs);\\
       r_n^1 := n(\Bs); & r_m^2 := m(\As);\\
    \end{array}\\
    \{ r_m^1 = r_m^2 \wedge r_n^1 = r_n^2 \wedge I(\sigma_1',\sigma_2') \}
\end{array}\]
Above we have created two copies of the object and a pre-relation that requires them to be initially equal and for some candidate condition $\varphi_m^n$ to hold. We then create a left program using one method order, and a right program with the other and the post-relation requires the return values agree and that some observational equivalence relation $I$ must hold.
 
A few problems come up with this attempted formulation.
First, starting from any initial state $\sigma_1$ is too strong; instead, we want to consider only reachable $\sigma_1$. In the interest of automation, we want tools to determine what's reachable after any finite number of invocations from a universal client and, consequently, we cannot (yet) state the pre-relation's additional constraint $\textsf{reachable}(\sigma_1)$.
Second, the post-condition above assumes we have already proved $I$ to be an observational equivalence relation. Again, in the interest of automation, we would like to \emph{prove} $I$ to be an equivalence relation, leading us to a different $k$-safety problem: $\{I\ \wedge \bar{v}_1\} r_1 = m(\bar{v}_1) \mid r_2 = m(\bar{v}_2) \{ r_1=r_2 \wedge I\}$. In fact, we need one of these quads for each $m$ in the ADT. 
So we now need $\textsf{card}(M) + 2$ proofs and, unfortunately, na\"{i}vely gluing them together is problematic for tools because the abstractions you need for the earlier proofs are different from those needed for later proofs.

Consequently, commutativity ..
escapes recent tools for $k$-safety~\cite{weaver,descartes}.
}

In this paper, we describe the first automatic method for verifying a given commutativity condition directly from the data structure's source code. 
At a first look, this problem seems difficult because data structures can be implemented in many ways and commutativity may seem to necessitate reasoning about full functional correctness.
%
Our first enabling insight is that, unlike elaborate invariants \& abstractions needed for verifying full functional correctness, we can achieve commutativity reasoning with specialized abstractions that are more targeted. Given two methods, we introduce an abstraction of states with respect to these methods that allows us to reason safely within the limited range of exploring the \emph{differences} between the behavior of pairs of operations when applied in either order and \emph{abstracting away} the state mutations that would be the same, regardless of the order in which they are applied. We call this an \mndifferencing\ abstraction.

Commutativity reasoning is  challenging also because we need to know whether two different concrete post states---one arising from $m(\Xs);n(\Ys)$ and one from $n(\Ys);m(\Xs)$---represent the same object state. Toward this challenge, we employ a notion of observational equivalence in our definition of commutativity, enabling us to  reason automatically about equivalent post-states, even in the absence of a data structure specification.
Crucially, in our observational equivalence relations, we exploit the convenience of using 
both abstract equality for some parts of the data structure (those that are order dependent such as the top few elements of a stack),  as well as direct equality for other parts of the data structure (other effects on, \eg, the remaining region of the stack, that are not order dependent).

Next, we introduce a novel automata-theoretic transformation for automating commutativity reasoning. Our transformation $\mathcal{E}(s_m,s_n)$ takes, as input, data structure implementations $s_m$ and $s_n$. The transformation  generates a symbolic automaton $\autencMN$ 
using a form of product construction to relate $s_m;s_n$ with $s_n;s_m$, 
but designed in a particular way so that, when it is provided
a candidate commutativity condition $\varphi_m^n$,
a safety proof on $\autencMN$ entails that $\varphi_m^n$ is a commutativity condition.
%
Our encoding co\"{e}rces a reachability analysis to 
perform the reasoning necessary for verifying commutativity: finding
\mndifferencing\ abstractions and proving observational equivalence.
%
For example, when applied to a Stack's \texttt{push}$(v)$ and \texttt{pop}
operations, a predicate analysis would discover predicates including whether the top value of the stack is equal to $v$ (as well as negations), and track where in the Stack implementation behaviors diverge on the basis of this abstraction. 


We implement our strategy in a new tool called \CityProver{}. It takes as input data structures written in a C-like language (with integers, arrays, and some pointers), implements the above described encoding and then employs Ultimate's~\cite{ultimate} or CPAchecker's~\cite{Beyer:CAV11} reachability \ignore{and termination } analyses to prove commutativity.
We show our approach to work on simple numeric data structures such as 
Counter, Memory, \textsf{SimpleSet},  \textsf{ArrayStack}, \textsf{ArrayQueue} and a hash table.

We consider the usability benefits of our strategy.
First, it can be applied to \emph{ad hoc} data structures for which the specification is not readily available.
Also, unlike data structure specifications, commutativity conditions can  be fairly compact, written as smaller logical formulae.
Consequently, with \CityProver{},  
a user can guess commutativity conditions and rely on our machinery to either prove them correct or find a counterexample; we discuss our use of this  methodology.

Our experiments confirm that our technique permits out-of-the-box reachability tools to verify some ADTs. They also reveal the current limitations of those tools to support ADTs in which equivalence  is up to \emph{permutations}. This leads to a blowup in the disjunctive reasoning that would be necessary, and an important question for future work.

\paragraph{Contributions.}
We present the first automated technique for verifying given commutativity conditions directly from data structure implementations.
\begin{itemize}
\item We formalize \emph{\mndifferencing\ abstractions} and observational equivalence relations for commutativity.
\item We present a symbolic technique for reducing commutativity verification to reachability.
\item We build a prototype tool \CityProver{}, which
generates a proof (or finds a counterexample) that $\varphi_m^n$ is a commutativity condition.
\item We demonstrate that \CityProver{} can prove commutativity of simple data structures including a memory cell, counter, two-place Set, array stack, array queue and hash table.
\end{itemize}

\noindent
\newif\ifarxiv
\arxivtrue
Our verified commutativity conditions can immediately be used with existing transactional object systems~\cite{aplas19}. Moreover, with some further research, they could be combined with linearizability proofs and used inside parallelizing compilers.

\paragraph{Limitations.}
In this paper we have focused on numeric programs. However, \mndifferencing\ abstractions and our reduction strategy appear to generalize to heap programs, left for future work.  
Also, while our reduction appears to be general, 
we were limited by the reasoning power of existing reachability tools, specifically, the need for permutation invariants.
Our work therefore establishes benchmarks for future work to improve those tools.



\lstset{numbers=none}
\begin{figure}[t]\centering
  \begin{tabular}{|l|l|}
    \hline
    \begin{minipage}[t]{2.8in}
\begin{lstlisting}[basicstyle=\sffamily]
class SimpleSet {
 private int a, b, sz;
 SimpleSet() { a=b=-1; sz=0; }
 void add(uint x) {
   if(a==-1 && b==-1) { a=x; sz++; }
   if(a!=-1 && b==-1) { b=x; sz++; }
   if(a==-1 && b!=-1) { a=x; sz++; }
   if(a!=-1 && b!=-1) { return; }
 }
 bool isin(uint y) { return (a==y||b==y);} 
 bool getsize() { return sz; } 
 void clear() { a=-1; b=-1; sz = 0; }
}
\end{lstlisting}
\end{minipage} \\
\hline
    \begin{minipage}[t]{2.8in}
\begin{lstlisting}[basicstyle=\sffamily]
class ArrayStack {
 private int A[MAX], top;
 ArrayStack() { top = -1; }
 bool push(int x) {
  if (top==MAX-1) return false;
  A[top++] = x; return true;
 }
 int pop() { 
  if (top == -1) return -1;
  else return A[top--]; }
 bool isempty() { return (top==-1); }
}
\end{lstlisting}
\end{minipage} \\
\hline                                                       
  \end{tabular}
  \caption{\label{fig:simpleset} (a) On the top, a \ttt{SimpleSet} data structure, capable of storing up to two
  non-zero identifiers (using private memory \ttt{a} and \ttt{b}) and tracking the size \ttt{sz} of the Set. (b) On the bottom, an \ttt{ArrayStack}
  data structure that implements a simple stack using an array and top index.}
\end{figure}

\subsection{Motivating Examples}

\newcommand\tta{\ttt{a}}
\newcommand\ttb{\ttt{b}}

Consider the \ttt{SimpleSet} data structure shown at the top of
Fig.~\ref{fig:simpleset}. This data structure is a simplification of a
Set, and is capable of storing up to two
natural numbers using private integers \ttt{a} and \ttt{b}.
Value $-1$ is reserved to indicate that
nothing is stored in the variable.
Method \addx\ checks to see if there is space available and,
if so, stores $x$ where space is available. 
%
Methods \isiny, \ttt{getsize}() and \ttt{clear}()
are straightforward.

A commutativity condition, written as a logical formula $\varphi_m^n$, describes the conditions under which two methods $m(\Xs)$ and $n(\Ys)$ commute, in terms of the argument values and the state of the data structure $\sigma$. 
%
Two methods \isinx{} and \isiny{} always commute because neither modifies the ADT,
so we say $\varphi_{\isinx}^{\isiny} \equiv \textsf{true}$.
The commutativity condition of methods \addx{} and
\isiny{} is more involved:
$$
\PCai \equiv x \neq y \vee (x=y\wedge \ttt{a} = x) \vee (x=y\wedge \ttt{b} = x)
$$
This condition specifies three situations (disjuncts) in which the
two operations commute. In the first case, the methods are operating
on different values. Method \isiny{} is a read-only operation and
since $y\neq x$, it is not affected by an attempt to insert $x$.
Moreover, regardless of the order of these methods, \addx{}
will either succeed or not (depending on whether space is available)
and this operation will not be affected by \isiny{}.
In the other disjuncts, the element being added is already in the Set, so
method invocations will observe the same return values regardless of the order and no changes (that could be observed by later methods) will be made by either of these methods.
%
Other commutativity conditions include:
$\varphi_{\isiny{}}^{\ttt{clear}} \equiv  (\aa\neq y \wedge \bb\neq y)$,
$\varphi_{\isiny{}}^{\ttt{getsize}} \equiv  \textsf{true}$,
$\varphi_{\addx{}}^{\ttt{clear}} \equiv  \textsf{false}$,
$\varphi_{\ttt{clear}}^{\ttt{getsize}} \equiv  \ttt{sz} = 0$ and
$\varphi_{\addx}^{\ttt{getsize}} \equiv  
  \ttt{a} = x \vee \ttt{b} = x \vee
  (\ttt{a} \neq x \wedge \ttt{a} \neq -1 \wedge \ttt{b} \neq x \wedge \ttt{b} \neq -1)
$.

While this example is somewhat artificial and small, we picked it to highlight couple of aspects.
One, commutativity reasoning can quickly become onerous to do manually.
Two, there can be multiple
concrete ways of representing the same semantic data structure state:
$\aa=5 \wedge \bb=3$ is the same as $\aa=3 \wedge \bb=5$.
This is typical of most data structure implementations,
 such as  list representations of unordered sets, hash tables,
 binary search trees, etc.
Our goal of commutativity reasoning from source code means that we do not have
the specification provided to us when
two concrete states are actually the same state semantically.
Below we will discuss how we overcome this challenge.  

As a second running example, let us consider an array based implementation of Stack, given at the bottom of Fig.~\ref{fig:simpleset}.
\ttt{ArrayStack} maintains  array \ttt{A} for data, a \topp\ index to indicate end of the stack, and has operations \ttt{push} and \ttt{pop}.
The commutativity condition $\PCpp \equiv \topp > -1 \wedge \ttt{A}[\topp] = x \wedge \topp < \ttt{MAX}$
 captures that they commute provided that there is at least
one element in the stack, the top value is the same as the value being pushed
and that there is enough space to \ttt{push}.




\subsection{Enabling abstractions}

As these examples show, specialized abstractions are needed for commutativity reasoning.
\ignore{
For example, when considering the commutativity 
$\varphi_{\ttt{wasFilled}}^{\ttt{clear}}$ 
of \ttt{wasFilled} and \ttt{clear}, 
for ensuring that return values agree, we can use an abstraction that only considers predicates
$\mathcal{P}_{\ttt{wasFilled}}^{\ttt{clear}} = 
\{\filled >0, \filled \leq 0\}$.
That is, we only need to consider how value of \filled\ is impacted by the two operations.
}
When considering the commutativity of \isiny{} and \ttt{clear}, we can use an abstraction that ignores \ttt{sz} and instead just reasons about \aa\ and \bb,
such as the predicates 
$ \aa=\ttt{y}$ and $ \bb=\ttt{y} $, along 
with their negations.
This abstraction not only ignores \ttt{sz}, but it also \emph{ignores other possible values} for \aa\ and \bb. The only relevant aspect of the state is whether or not \ttt{y} is in the set.
For \ttt{ArrayStack}, when considering \pushx{} and \ttt{pop}, we can similarly
abstract away deeper parts of the stack: for determining \emph{return values}, our abstraction only needs to consider the top value.
In Section~\ref{sec:abstraction}, we will formalize this concept as an \emph{\mndifferencing\ abstraction} denoted $(\alpha_m^n, R_\alpha)$ and, in Section~\ref{sec:verification}, describe a strategy that discovers these abstractions automatically. This concept is illustrated on the left of the following diagram:
\begin{center}
  \includegraphics[width=3.3in]{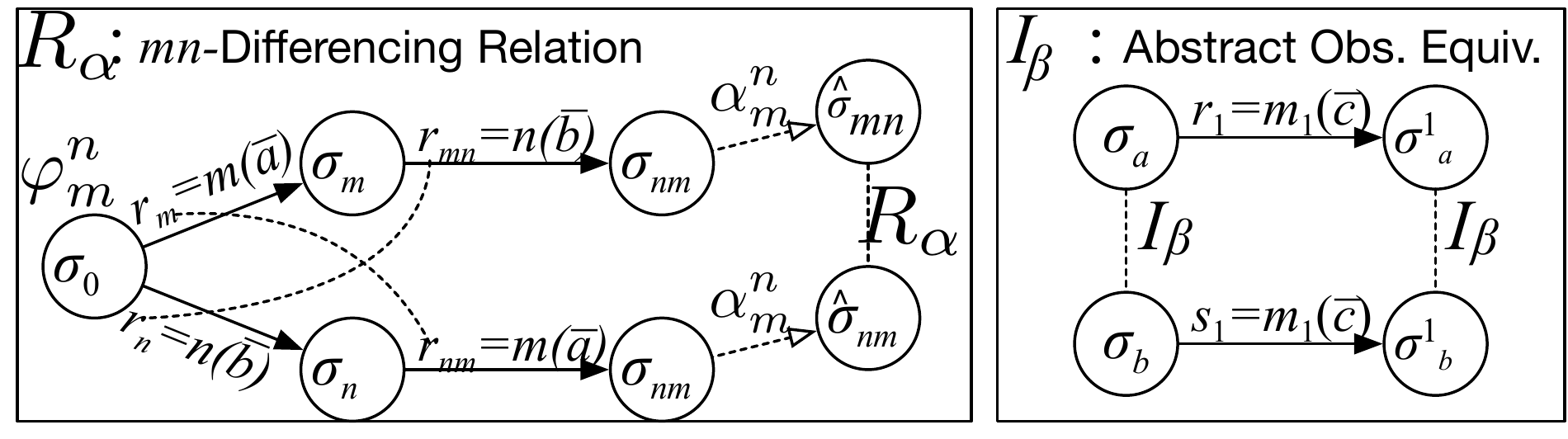}
\end{center}
Intuitively, $R_\alpha$ is a relation on abstract states whose purpose is to ``summarize'' the possible pairs of post-states that (i) originate from
a state $\sigma_0$ satisfying the commutativity condition $\varphi_m^n$ and (ii) will have agreed on
intermediate return values along the way.
To this end, the abstraction $\alpha_m^n$ must be fine-grained enough to reason about agreement on these intermediate 
return values ($r_m = r_{nm}$, $r_n = r_{mn}$) but can abstract away other details. 
In the \isiny{}/\ttt{clear} example, the proposed abstraction that simply uses predicates $\ttt{a} = y$ and $\ttt{b} = y$ is an \mndifferencing\ abstraction because it is all that's needed to show that, in either order, the return values agree.


\paragraph{Observational Equivalence.}
While this \mndifferencing\ abstraction ensures that return values agree, how can we ensure that either order leads to equivalent post-states? 
First, for some start state $\sigma$, let us denote by $\sigma_{mn}$ the state arising by applying $m(\As);n(\Bs)$. We similarly define $\sigma_{nm}$. We define observational equivalence in a standard way, by requiring that any sequence of method invocations (or \emph{actions}) $m_1(\As_1)/r_1;m_1(\As_1)/r_2;\ldots$, when separately applied to $\sigma_{mn}$ and $\sigma_{nm}$ returns sequences of values that agree. This can be seen on the right in the above diagram: relation $I_\beta$ is maintained.

Observational equivalence lets us balance abstraction (for reasoning about $m,n$-specific interactions) with  equality (for reasoning about parts of the state that are unaffected by the order of $m$ and $n$). 
That is, observational equivalence is here used in the context of two states ($\sigma_{mn}$ and $\sigma_{nm}$) that
originate from precisely the same state\footnote{Note that it is sufficient to start from same concrete state, rather than two observationally equivalent start states in the reasoning, as only one of the order of the methods execution will occur at run-time.}, and have likely not strayed too far from each other. In this way, the infinite extension often collapses to comparing two states that are exactly equal.
%
For example, consider the relations:
\[\begin{array}{rl}
I_{AS}(\sigma_{mn},\sigma_{nm}) \equiv & \topp_{mn} = \topp_{nm}  \wedge\\
& (\forall i. 0 \leq i \leq \topp_{mn} \Rightarrow \aa_{mn}[i] = \aa_{nm}[i]) \\ 
I_{SS}(\sigma_{mn},\sigma_{nm}) \equiv & ((\aa_{mn} = \aa_{nm} \wedge \bb_{mn} = \bb_{nm}) \\
& \;\;\;\;\;\; \vee\; (\aa_{mn} = \bb_{nm} \wedge \bb_{mn} = \aa_{nm})) \\ 
& \;\;\wedge\;\; 	(\ttt{sz}_{mn} = \ttt{sz}_{nm})\\
\end{array}\]
For the \ttt{ArrayStack}, $I_{AS}$ says that the two states agree on the (ordered) values in the Stack. ($\topp_{mn}$ means the value of $\topp$ in $\sigma_{mn}$.)
For \ttt{SimpleSet},  $I_{SS}$ specifies that two states are equivalent provided that they are storing the same values---perhaps in different ways---and they agree on the size.

If we can find an $I$ and a valid \mndifferencing\ relation $R_\alpha$ such that $R_\alpha \Rightarrow I$, then we have proved commutativity.
These abstractions work well together when used in tools (discussed below). When considering, for example, \ttt{push} and \ttt{pop}, only the top element of the stack is changed differently depending on the order, and our abstraction lets us treat the other aspects of the state (that are unchanged or changed in the same way regardless of $m-n$ order) as being exactly the same.
%
%
This observation let us employ verification analysis in a way that is focused, and avoids the need for showing full functional correctness.
%

\subsection{Challenges}

Above we have sketched some of the enabling insights behind our work, but many questions remain. In the rest of this paper, we answer questions such as:

\begin{itemize*}
\item Is there a formal equivalence between proving commutativity and finding sufficient \mndifferencing\ abstractions and observational equivalence? (Section~\ref{sec:abstraction})
\item Can we employ these concepts to reduce commutativity verification to reachability, so that \mndifferencing\ and observational equivalence abstractions can be correlated with non-reachability proofs?
 (Section~\ref{sec:verification})
\item Can we build a tool around these concepts? (Section~\ref{sec:impl})
\item Can we leverage existing reachability solvers to prove commutativity of simple ADTs? (Section~\ref{sec:eval})
\item How usable is our technique? (Section~\ref{subsec:usability})
\end{itemize*}


\ignore{
In Section~\ref{sec:verification}, we describe an algorithm for
automatically verifying a given commutativity condition $\varphi_m^n$.
Our technique reduces this problem to  reachability via a transformation
$
    \mathcal{E}(M, s_m, s_n) = \ignore{\lambda \varphi_m^n.\ }\autencMmn{ \varphi_m^n }
$
from the data structure
method implementations $M = \{s_m, s_n, \ldots\}$, to an \emph{encoding}
$\autencMmn{ \varphi_m^n }$. This resulting
$\autencMmn{ \varphi_m^n }$ consumes a formula
$\varphi_m^n$ and is designed such that  its safety
implies that $\varphi_m^n$ is a commutativity condition.
Overall, the encoding co\"{e}rces program analysis tools to automatically discover \mndifferencing\ abstractions and observational equivalence relations in order to prove commutativity.
The main soundness result in Section~\ref{sec:verification} is that a proof that
$\autencMN$ cannot reach \ttt{ERROR} implies that $\varphi_m^n$ is a 
commutativity condition for $m$ and $n$.

\subsection{Proving Commutativity with \CityProver}

In Section~\ref{sec:impl}, we describe an implementation of our technique which takes as input a candidate commutativity condition and a data structure implementation written in C. Our implementation   generates the encoding described above and then applies reachability analyses to the encoding in order to verify the commutativity condition.
%
}

\section{Preliminaries}

\paragraph{Language.} We work
with a simple model of a (sequential) object-oriented language.
We will denote an object by $o$. Objects can have 
member fields $o.a$ and, for the purposes of this paper, we
assume them to be integers,  structs or integer arrays. 
Methods are denoted $o.m(\bar{x}),o.n(\bar{y}),\ldots$ where $\bar{x}$ is
a vector of the arguments. We often will omit the $o$ when it is unneeded; we
use the notation 
$m(\Xs)/\Rms$ to refer to the return variables $\Rms$. 
We use $\As$ to denote a vector of argument \emph{values}, 
$\Us$ to denote a vector of return \emph{values}
and $m(\As)/\Us$ to denote a corresponding invocation of a method which we  call an \emph{action}. For a method $n$, our convention will be
the signature $n(\Ys)/\Rns$ and values $n(\Bs)/\Vs$.
%
Methods source code is parsed from C into control-flow automata (CFA)~\cite{Henzinger2002}, discussed in Section~\ref{subsec:objects}, using \texttt{assume} to represent branching and loops. Edges are labeled with straight-line ASTs:
\[\begin{array}{rclrcl}
  e &::=& o.a \mid c \mid \mathbb{Z} \mid x  \mid e \otimes e,  
  \qquad b ::= b \oplus b \mid \neg b \mid \textsf{true},  \\
  s &::=& \texttt{assume(} b \texttt{)} \mid s_1 \texttt{ ; } s_2 \mid x \texttt{ := } e 
\end{array}\]
Above $\otimes$ represents typical arithmetic operations on integers and 
$\oplus$ represents typical boolean relations.
We use $s_m$ to informally refer to the source code of object method $m$ (which is, formally, represented as an object CFA, discussed in Section~\ref{subsec:objects}).
For simplicity, we assume that one object method cannot call another, and that 
all object methods terminate. Non-terminating object methods are typically not useful and their termination can 
be confirmed using existing termination tools (\eg~\cite{T2}).

We fix a single object $o$, denote that object's concrete state space $\Sigma$, and assume decidable
equality.
We denote $\sigma \xrightarrow{m(\As)/\Us} \sigma'$
for the big-step semantics in which the arguments are provided, the entire method is reduced. For lack of space, we omit the small-step
semantics $\sem{s}$ of individual statements. 
For the big-step semantics, we assume that such a successor state $\sigma'$ is always defined (total) and is unique (determinism). Programs can be transformed so these conditions hold, via wrapping~\cite{tacas18} and prophecy variables~\cite{abadilamport}, resp.

\begin{definition}[Observational Equivalence (\eg~\cite{Koskinen2015})]
We define relation $\opeq \subseteq  \Sigma \times \Sigma$ as the following greatest fixpoint:
  $$
  \infer={\sigma_1 \opeq \sigma_2}
         { \forall m(\As) \in M.\ \sigma_1 \xrightarrow{ m(\As)/\Rms } \sigma_1'
           & \sigma_2 \xrightarrow{ m(\As)/\Rns } \sigma_2'
           & \Rms = \Rns
           & \sigma'_1 \opeq \sigma'_2}
$$
\end{definition}
\noindent
The above co-inductive definition expresses that two states $\sigma_1$ and $\sigma_2$ of an object are observationally equivalent $\opeq$ provided that, when any given method invocation $m(\As)$ is applied to both $\sigma_1$ and $\sigma_2$, then the respective return values agree. Moreover, the resulting post-states maintain the $\opeq$ relation.
Notice that a counterexample to observational equivalence (\ie, something that is not in the relation)
is a finite sequence of method operations $m_1(\As_1),...,m_k(\As_k)$ applied to both $\sigma_1$ and
$\sigma_2$ such that for $m_k(\As_k)$,  $\Rms_k \neq \Rns_k$. 

We next use observational equivalence to define commutativity.
As is typical~\cite{DBLP:conf/pldi/DimitrovRVK14,tacas18} we define commutativity first at the layer of an action, which are particular \emph{values}, and second at the layer of a method, which includes a quantification over all of the possible values for the arguments and return variables.

\begin{definition}[Commutativity]\label{def:commute} For values $\As,\Bs$, return values $\Us,\Vs$, and methods $m$ and $n$,
``{\bf actions  $o.m(\As)/\Us $ and $o.n(\Bs)/\Vs$ commute},'' denoted $o.m(\As)/\Us \bowtie o.n(\Bs)/\Vs$, if
\[\begin{array}{c}
   \forall\sigma.\  \sigma \xrightarrow{m(\As)/\Us} \sigma_m \xrightarrow{n(\Bs)/\Vs} \sigma_{mn}
     \;\wedge\;
      \sigma \xrightarrow{n(\Bs)/\Vs} \sigma_n \xrightarrow{m(\As)/\Us} \sigma_{nm}
\;\;  \\
   \Rightarrow \;\; \sigma_{mn} \simeq \sigma_{nm}
\end{array}\]
(Note that these are \emph{values}, so action commutativity requires return value agreement.) 
We further define ``{\bf methods $o.m$ and $o.n$ commute},'' denoted $o.m\ \bowtie\ o.n$ provided that\\
$ \forall \As\ \Bs\ \Us\ \Vs. o.m(\As)/\Us \bowtie o.n(\Bs)/\Vs$.
\end{definition}
\noindent
The quantification $\forall  \As\ \Bs\ \Us\ \Vs$ above means  
vectors of all possible argument and return values. Note that commutativity begins
with a single state $\sigma$, but implies a \emph{relation} on states. (We
will come back to this later.)
Our work extends to a more
fine-grained notion of commutativity: an asymmetric version called
left-movers and right-movers~\cite{lmrm}, where a method commutes in
one direction and not the other.


We will work with commutativity conditions for methods
$m$ and $n$ as logical
formulae over initial states and the arguments/return values of the methods. We denote a logical commutativity formula as $\varphi_m^n$ and assume a decidable
interpretation of formulae:
$\sem{\varphi_m^n} : (\sigma,\Xs,\Ys,\Rms,\Rns) \rightarrow \mathbb{B}$.
%
(We tuple the arguments for brevity.)
The first argument is the initial state.  Commutativity \emph{post}-
and \emph{mid}-conditions can also be written~\cite{KR:PLDI11} but
here, for simplicity, we focus on commutativity \emph{pre}-conditions.
We may write $\sem{\varphi_m^n}$ as 
$\varphi_m^n$ when it is clear from context that $\varphi_m^n$ is meant to be interpreted.

\begin{definition}[Commutativity Condition]
We say that logical formula $\varphi_m^n$ is a commutativity condition for $m$ and $n$
provided that 
$\forall \sigma\ \As\ \Bs\ \Us\ \Vs.\
\sem{\varphi_m^n}\ \sigma\ \As\ \Bs\ \Us\ \Vs
\Rightarrow m(\As)/\Us\ \bowtie\ n(\Bs)/\Vs$.
\end{definition}

\newcommand\enct{\mathcal{E}}
\newcommand\autencD{\mathcal{A}^\mathcal{E}}
\newcommand\encQ{Q_{\enct}}
\newcommand\encqO{q^0_{\enct}}
\newcommand\encX{X_{\enct}}
\newcommand\encOps{\Ops_{\enct}}
\newcommand\encarr{\Ecfarightarrow{\,}}
\newcommand\commu{H_{c}}
\newcommand\ncommu{H_{nc}}

\section{Abstraction}
\label{sec:abstraction}

We now formalize abstractions
for two interconnected aspects of commutativity reasoning:
\mndifferencing\ and observational equivalence relations.
%
These treatments provide a route to automation and re-use
of abstraction synthesis techniques, discussed in Section~\ref{sec:verification}.


\newcommand\rms{\bar{v}_{m}}
\newcommand\rmns{\bar{v}_{mn}}
\newcommand\rns{\bar{v}_{n}}
\newcommand\rnms{\bar{v}_{nm}}
For convenience, we define post-states \textsf{posts} and return value agreement \textsf{rvsagree} as follows:
\[\begin{array}{l}
   \textsf{posts}(\sigma, m, \As, n, \Bs)
   \equiv (\sigma_{mn},\sigma_{nm}) \textrm{ such that }\\
\;\;\; \;\;\sigma \xrightarrow{m(\As)/\rms} \sigma_m
    \wedge \sigma_m \xrightarrow{n(\Bs)/\rmns} \sigma_{mn}\\
\;\;\;\;\;\;\;\;\;\; \wedge \; \sigma \xrightarrow{n(\Bs)/\rns} \sigma_{n}
       \wedge \sigma_n \xrightarrow{m(\As)/\rnms} \sigma_{nm}\\
   \textsf{rvsagree}(\sigma, m, \As, n, \Bs)
   \equiv \textrm{for } \rms,\rmns,\rns,\rnms \textrm{ then  }\\
\;\;\; \;\;\sigma \xrightarrow{m(\As)/\rms} \sigma_m
    \wedge \sigma_m \xrightarrow{n(\Bs)/\rns} \sigma_{mn}\\
\;\;\;\;\;\;\;\;\;\; \wedge \; \sigma \xrightarrow{n(\Bs)/\rns} \sigma_{n}
       \wedge \sigma_n \xrightarrow{m(\As)/\rnms} \sigma_{nm},\\
\;\;\;\;\;\; \textrm{implies } \rms=\rnms \textrm{ and }\rns=\rmns 
   \end{array}\]

\newcommand{\SigmaA}{\Sigma^\alpha}
\newcommand{\SigmaB}{\Sigma^\beta}

\noindent
We now formalize \mndifferencing\ abstraction:
%
%
%
\begin{definition}[\mndifferencing\ Abstraction $(\alpha,R_\alpha)$]
\label{defn:alphaabs}
Let $o$ be an object with state space $\Sigma$, and consider two methods $m$ and $n$.
Let $\alpha : \Sigma \rightarrow \SigmaA$ be an abstraction of the states, and $\gamma: \SigmaA \to \mathcal{P}(\Sigma)$ the corresponding concretization.
Let $\sem{R_\alpha} : \SigmaA \rightarrow \SigmaA \rightarrow \mathbb{B}$ be a relation on abstract states.
 We say that $(\alpha, R_\alpha)$ is an \mndifferencing\ abstraction if
\[\begin{array}{ll}
  \forall \sigma_1^\alpha, \sigma_2^\alpha \in \SigmaA.
  \sem{R_\alpha}(\sigma_1^\alpha, \sigma_2^\alpha) \Rightarrow\\
  \qquad \forall \sigma \As \Bs.\ \textsf{posts}(\sigma, m, \As, n, \Bs) \in \gamma(\sigma_1^\alpha) \times \gamma(\sigma_2^\alpha) \Rightarrow \\
  \qquad
  \textsf{rvsagree}(\sigma, m, \As, n, \Bs)
\end{array}\]
\end{definition}
%
\noindent
Here $\alpha$ provides an abstraction and $R_\alpha$ is a relation in that
abstract domain. Intuitively, $R_\alpha$ relates pairs of post-states that (i) originate from same state $\sigma$ and (ii) agree on intermediate return values. 
$\alpha$ must be a precise enough abstraction so that $R_\alpha$
can discriminate between pairs of post-states where return values will have agreed versus disagreed.
For the \isinx{}/\ttt{clear} example, we can let $\alpha$ be the abstraction that 
tracks whether $\aa = x$ and whether $\bb =x$. Then
$R_\alpha(\sigma_1,\sigma_2) \equiv (\aa = x)_1 = (\aa= x)_2 \wedge (\bb = x)_1 = (\bb= x)_2$,
\ie\ the relation that tracks if $\sigma_1$ and $\sigma_2$  agree on whether $x$ is in the set.
\begin{definition}
\label{defn:impliesA}
  Let $(\alpha, R_\alpha)$ be an \mndifferencing\ abstraction
   and $\varphi_m^n$ a logical formula on concrete states and actions of $m$ and $n$.
  We say that ``$\varphi_m^n$ \emph{implies} $(\alpha,R_\alpha)$'' if
\[
  \forall \sigma\ \As\ \Bs\ \Rms\ \Rns.\ \varphi_m^n(\sigma, \As, \Bs, \Rms, \Rns) \Rightarrow R_\alpha(\alpha(\sigma_{mn}),\alpha(\sigma_{nm}))
\]
where $(\sigma_{mn},\sigma_{nm}) = \textsf{posts}(\sigma, m, \As, n, \Bs)$.
\end{definition}
\noindent
The above definition treats commutativity condition $\varphi_m^n$ as a pre-condition, specifically one that characterizes certain start states $\sigma_0$ for which the \mndifferencing\ abstraction will hold.
If we let $\varphi_{\isinx}^{\ttt{clear}} \equiv \aa \neq x \wedge \bb \neq x$, this will imply $R_\alpha$ in the \textsf{posts}: the post states will agree on whether $x$ is in the set, thus capturing our intuitive understanding of commutativity.

Next, we introduce another (possibly different) abstraction which helps to reason about observational equivalence of the post-states reached.
\begin{definition}
\label{defn:betaabs}
Let $\beta : \Sigma \rightarrow \SigmaB$ be an abstraction function, 
with corresponding concretization function $\delta : \SigmaB \rightarrow \mathcal{P}(\Sigma) $, and 
let $I_\beta$ be a relation on these abstract states, $\sem{I_\beta} : \SigmaB \times \SigmaB \rightarrow \mathbb{B}$. Then $I_\beta$ is an \emph{observational equivalence relation} iff:
\[\begin{array}{l}
\forall \sigma_1^\beta, \sigma_2^\beta \in \SigmaB.\ \sem{I_\beta} (\sigma_1^\beta, \sigma_2^\beta) \Rightarrow\\
\;\;\;\forall \sigma_1 \in \delta(\sigma_1^\beta), \sigma_2 \in \delta(\sigma_2^\beta).
\; \sigma_1 \opeq \sigma_2
\end{array}\]
\end{definition}
\noindent
$I_{SS}$,  defined earlier, is such a relation.

The next definition is a natural notion stating when an abstraction is coarser than the other; we use it later.
\begin{definition}
\label{defn:impliesB}
For all $(\alpha,R_\alpha)$ and $(\beta,I_\beta)$ 
we say that ``$(\alpha,R_\alpha)$ implies $(\beta,I_\beta)$'' iff:
\[
  \forall \sigma_1, \sigma_2 \in \Sigma. \sem{R_\alpha}(\alpha(\sigma_1), \alpha(\sigma_2)) \Rightarrow \sem{I_\beta}(\beta(\sigma_1), \beta(\sigma_2))
\]
\end{definition}

\noindent
The following theorem provides sufficient conditions for a commutativity condition for the two methods, with respect to these abstractions.
\begin{theorem}
  \label{thm:abstraction}
Let $o$ be an object with state space $\Sigma$, and methods $m(\Xs)$ and $n(\Ys)$. 
Let $\varphi_m^n$ be a logical formula on $\Sigma$ and actions of $m$ and $n$.
Let $(\alpha_m^n, R_m^n)$ be an \mndifferencing\ abstraction and $(\beta,I_\beta)$ such that $I_\beta$ is an abstract observational equivalence relation.
If $\varphi_m^n$ implies $(\alpha_m^n, R_m^n)$ and $(\alpha_m^n, R_m^n)$ implies $(\beta, I_\beta)$ then $\varphi_m^n$ is a commutativity condition for $m$ and $n$.
\begin{proof}[Proof sketch]
Fix $\sigma \in \Sigma$, and $m(\As)/\Rms$ and $n(\Bs)/\Rns$ actions for $m$ and $n$ respectively. We need to show that $m(\As)/\Rms \bowtie n(\Bs)/\Rns$. In particular, we need to show that the return values agree, and the post states reached with methods commuted are observationally equivalent. From Definition \ref{defn:impliesA}, we have that $R$ holds for the $\alpha$-abstraction of post states, and then from Definition \ref{defn:alphaabs} it follows that the return values agree. On the other hand, from Definition \ref{defn:impliesB} it follows that $I_\beta$ holds for the $\beta$-abstraction of post states as well, and from Definition \ref{defn:betaabs} it follows that the (concrete) post states are observationally equivalent.

\end{proof}
\end{theorem}
\noindent
The idea is that for a $m$/$n$ pair, find an $(\alpha, R_\alpha)$ that witnesses the way in which the order causes divergence but strong enough to imply \emph{some} abstract equivalence relation $(\beta,I_\beta)$. 
Given that $\varphi_{\isinx}^{\ttt{clear}}$ implies the above $R_\alpha$, and that 
it is easy to see that $R_\alpha$ implies $I_{SS}$, we can conclude that 
$\varphi_{\isinx}^{\ttt{clear}}$ is a valid commutativity condition.

\section{Reduction to Reachability}
\label{sec:verification}

We now describe how we reduce the task of verifying commutativity condition $\varphi_m^n$
to a reachability problem.
To this end, we need a representation of object implementations, as well
as the output encoding. As noted, we build on the well-established notion of control-flow
automata (CFA)~\cite{Henzinger2002}, extending them slightly to represent objects. We then describe
our transformation from an input object CFA to an output encoding CFA $\autencMN$ with an error state $q_{er}$.  Finally,
we prove that, if $q_{er}$ is unreachable in $\autencMN$, then
$\varphi_m^n$ is a valid commutativity condition.


\subsection{Object Implementations}
\label{subsec:objects}


\begin{definition}[Control-flow automaton~\cite{Henzinger2002}]
A (deterministic) \emph{control flow automaton} is a
tuple $\aut = \langle Q, q_0, X, s, \cfarightarrow{\,} \rangle$ where $Q$ is a
finite set of control locations and $q_0$ is the initial control location, $X$
is a finite sets of typed variables, $s$ is the loop/branch-free statement language (as defined
earlier) and $\cfarightarrow{\,}\subseteq Q \times s \times Q$ is a finite set
of labeled edges.
\end{definition}

\newcommand\run{r}

\paragraph{Valuations and semantics.}
We define a \emph{valuation} of variables 
$\vln : X \rightarrow \mathit{Val}$ to be a mapping from
variable names to values. Let $\Vlns$ be the set of all valuations.
The notation $\vln' \in \sem{s}\vln$ means that executing statement $s$, using
the values given in $\vln$, leads to a new
valuation $\vln'$, mapping variables $X$ to new values. Notations $\sem{e}\vln$ and $\sem{b}\vln$ represent side-effect free numeric and boolean valuations, respectively.
We assume that for every $\vln,s$, that $\sem{s}\vln$ is computed in finite time.
The automaton edges, along with $\sem{s}$, give rise to possible transitions, denoted
$(q,\vln) \cfarightarrow{s} (q',\vln')$ but omit these rules for lack of space.

A \emph{run} of a CFA is an alternation of automaton states and
valuations denoted $\run=q_0,\vln_0,q_1,\vln_1,q_2,\dots$
such that $\forall i \geq0.\ (q_i,\vln_i) \cfarightarrow{s} (q_{i+1},\vln_{i+1})$.
%
We say  $\aut$ can \emph{reach} automaton state $q$ (\emph{safety}) provided 
there exists a run $\run= q_0,\vln_0,q_1,\vln_1,...$ such that there is
some $i \geq 0 $ such that  $q_i=q$.
%
We next conservatively extend CFAs to represent data structure implementations:

\begin{definition}[Object CFAs]
  An  \emph{object control flow automaton} for object $o$
  with methods $M=\{m_1,...m_k\}$, is:
\[\begin{array}{l}
\aut_o = \langle Q_o, [q^\textsf{init}_0,q^\textsf{clone}_0,q^{m_1}_0,\ldots,q^{m_k}_0], X_o, s, \cfarightarrow{\,} \rangle\\
X_o = X^{st} \cup \{ \textsf{this}_o \} \cup X^\textsf{init} \bigcup_{i\in[1,k]}  (X^{m_i} \cup X^{m_i}_{\bar{x}_i} \cup X^{m_i}_{\Rms_i}) 
\end{array}\]
where $Q_o$ is a finite set of control locations,
$q_0^\textsf{init}$ is the initial control location of the initialization routine,
$q_0^\textsf{clone}$ is the initial control location of the clone routine,
 and $q^{m_i}_0$ is the initial
control location for method $m_i$. 
Component $X_o$
is a union of 
sets of typed variables:
$X^{st}$ representing the object state/fields,
a \textsf{this} variable, and
$X^\textsf{init}$ representing  the initialization routine's local variables.
For each $m_i$,
$X^{m_i}$ represents method $m_i$'s local variables,
$X_{\bar{x}_i}^{m_i}$ represents parameters, and
$X^{m_i}_{\Rms_i}$ represents return variables.
%
Finally, $s$ is the statement language and
 $\cfarightarrow{\,}\subseteq Q_o \times s \times Q_o$ is
a finite set of labeled edges.
\end{definition}

\noindent
For simplicity, the above definition avoids inter-procedural reasoning. However,
we need to be able to invoke methods. 
Above, we will call each $q_0^{m_i}$ node
the \emph{entry node} for the implementation of method $m_i$ and we additionally require that, for every method,
there is a special \emph{exit node} $q_{ex}^{m_i}$. We require that
the edges that lead to $q_{ex}^{m_i}$ contain \texttt{return}($\bar{v}$) statements.
%
Arcs are required to be deterministic. This is not without loss
of generality. We support nondeterminism in our examples
by symbolically determinizing the input
CFA: whenever there is a nondeterministic operation $m(\Xs)$, we
can augment $\Xs$ with a fresh prophecy variable $\rho$, and replace
the appropriate statements with a version that consults $\rho$ to resolve nondeterminism
(see~\cite{DBLP:conf/popl/CookK11}).
%
%
The semantics $(q,\vln)\cfarightarrow{s}(q',\vln')$ of $\aut_o$ induce a labeled transition system, with state space $\Sigma_{\aut_o} = Q_o\times \Theta$.
%
Naturally, commutativity of an object CFA is defined in terms of this induced transition system.

%

%

\subsection{Transformation}
\label{subsec:transformation}
%
%

\newcommand\enc{\mathcal{E}}
\newcommand\earrow[1]{\cfarightarrow{#1}_\enc}

\newcommand\inl[1]{-\!\!-\!\textsf{inl}(#1)\!\!\earrow{\,}}
\newcommand\nds{\bar{*}}
\newcommand\vrm{\bar{r}_{m}}
\newcommand\vrn{\bar{r}_{n}}
\newcommand\vrmn{\bar{r}_{mn}}
\newcommand\vrnm{\bar{r}_{nm}}

We now  define the transformation.
First,  syntactic sugar:
\[\begin{array}{l}
q \inl{m_i,o,\Xs,\Rms} q' \equiv\\
\;\;\;\{ q \earrow{\bar{x}_i := \Xs;} q_0^{m_i}, \;\;
   q_{ex}^{m_i} \earrow{\Rms := \bar{r}_i;} q' \}
\end{array}\]
This definition allow us to emulate function calls 
to a method $m_i$, starting from CFA node $q$. Values
$\Xs$ are provided as arguments, and arcs are created 
to the entry node $q_0^{m_i}$ for method $m_i$.
Furthermore, return values are saved into $\Rms$ and 
an arc is created from the exit node $ q_{ex}^{m_i} $ to $q'$.
Also, we will let $assume(\bar{x} \neq \bar{y})$ mean the disjunction of inequality between corresponding vector elements, \ie\, if there are $N$ elements, then the notation means $x_0\neq y_0 \vee\cdots\vee x_N\neq y_N$.

\newcommand\cmt[1]{\hspace{-0.1in}\textrm{(#1)}}

\begin{definition}[Transformation]
For an input object CFA
$\aut_o = \langle Q_o, [q^{c}_0, q^\textsf{init}_0,q^{m_1}_0,\ldots,q^{m_k}_0], X_o, s, \cfarightarrow{\,} \rangle$,
the result of the transformation when applied to methods
$m(\Xs),n(\Ys)$, is output CFA
$ \autencMmn{\varphi_m^n} = \langle Q_\enc,q_0^\enc,X_\enc,s_\enc,\earrow{\ } \rangle$, where $\earrow{\,} \equiv$ is the union of
\begingroup
\allowdisplaybreaks
\begin{align*}
&\cfarightarrow{\,} & \cmt{$o$'s source}\\
&q^\enc_0 \inl{\ttt{init},nil,[],[o_1]} q_1 & \cmt{Construct.}\\
&\cup_{m_i} \left\{ \begin{array}{l}
    q_1 \earrow{\Xs := \nds;} q_{1i}, \\
    q_{1i} \inl{m,o_1,\Xs,nil} q_1 \end{array}\right. & \cmt{Reachbl.~$o_1$}\\
&q_1 \earrow{\As := \nds;\; \Bs := \nds;\; assume(\varphi_m^n(o_1,\As,\Bs))} q_{11}
   & \cmt{Assume $\varphi_m^n$}\\
&q_{11} \inl{\ttt{clone},o_1,[],[o_2]} q_2 &\cmt{Clone $o_1$)}\\
&q_2 \inl{m,o_1,\As,\vrm} q_{21} & \cmt{$m;n$ for $o_1$} \\
&q_{21} \inl{n,o_1,\Bs,\vrmn} q_{22} \\
&q_{22} \inl{n,o_2,\Bs,\vrn} q_{23} & \cmt{$n;m$ for $o_2$}\\
&q_{23} \inl{m,o_2,\As,\vrnm} q_{3} \\
&q_3 \earrow{assume(\vrm \neq \vrnm \vee \vrn \neq \vrmn)} q_{er} & \cmt{Rv's  agree}\\
&q_3 \earrow{\,} q_4 & \cmt{Loop}\\
&\cup_{m_i} \left\{\begin{array}{l}
   q_4 \earrow{ \As := \nds } q_{4i} \\
   q_{4i} \inl{m_i,o_1,\As,\Rms} q_{4i'} \\
   q_{4i'} \inl{m_i,o_2,\As,\Rns} q_{4i''} \\
   q_{4i''} \earrow{assume(\Rms \neq \Rns)} q_{er} \\
   q_{4i''} \earrow{\,} q_4
   \end{array}\right. & \cmt{Any $m_i$}
\end{align*}
\endgroup
and $Q_\enc$ is the union of $Q_o$ and all CFA nodes above,
$s_\enc$ is the union of $s$ and all additional statements above, and
$X_\enc = X_o \cup \{ o_1, o_2, \As, \Bs, \vrm, \vrn, \vrmn, \vrnm, nil, \Rms, \Rns\}$.
\end{definition}

We now describe the above transformation intuitively.
Node $q^\enc_0$ is the initial
node of the automaton. 
The transformation employs the implementation \emph{source code}
of the data structure CFA, given by $\cfarightarrow{\,}$.
The key theorem below says that non-reachability of
$q_{\mathit{er}}$ entails that $\varphi_m^n$ is a commutativity condition. We now discuss how the components of 
$\autencMN$ employ this strategy: 
\begin{enumerate*}
\item The first arcs of $Q_\enc$ are designed so that, for any
  reachable state $\sigma$ of object $o$, there will be some run of $\autencMN$
that witnesses that state. This is accomplished by arcs from $q_1$ into 
the entry node of each possible $m_i$, first letting $Q_\enc$ 
nondeterminsitically set arguments $\Xs$.
\item From $q_1$, nondeterministic choices are made for the method arguments
$m(\As)$ and $n(\Bs$), and then {\bf candidate condition} $\varphi_m^n$
is assumed.
And arc then causes a run to clone $o_1$ to $o_2$, and then
invoke $m(\As);n(\Bs)$ on $o_1$ and $n(\Bs);m(\As)$ on $o_2$.
\item From $q_3$, there is an arc to $q_{er}$ which is feasible
if the return values disagree. There is also an arc to $q_4$.
\item $q_4$ is the start of a loop. From $q_4$, for any possible method 
$m_i$, there is an arc with a statement $\As := \nds$ to chose nondeterministic
values, and then invoke $m_i(\As)$ on both $o_1$ and $o_2$.
If it is possible for the resulting return values to disagree, then a run
could proceed to $q_{er}$.
\end{enumerate*}

\paragraph{Proving that $q_{er}$ is unreachable.}
From the structure of $Q_\enc$, there are two ways in which $q_{er}$ could be reached
and a proof of non-reachability involves establishing invariants at two locations:

\begin{itemize*}
\item \emph{\mndifferencing\ relation at $q_3$}: At this location a program analysis
tool must infer a condition $R_\alpha$ that is strong enough to avoid $q_{er}$.
However, it may leverage the fact that $o_1$ and $o_2$ are similar.

\item \emph{Observational equivalence at $q_4$}: At this location,
$\autencMN$ considers any sequence of
    method calls $m',m'',\ldots$ that could be applied to both $o_1$
    and $o_2$. If observational equivalence does not hold, then
    there will be a run of $\autencMN$ that applies that
    some sequence of method calls to $o_1$ and $o_2$, eventually finding a discrepancy 
    in return values and making a transition to $q_{\mathit{er}}$.
    To show this is impossible, a program analysis tool can use an
   observational equivalence relation $I(o_1,o_2)$.    

\item \emph{Implication}: Finally, if it can be shown that $R_\alpha$ implies $I$,
 then $\varphi_m^n$ is a valid commutativity condition.

\end{itemize*}

%
%
%
%

\noindent
The utility of our encoding can 
be seen when considering what happens when a reachability analysis tool is applied to $\autencMmn{\varphi_m^n}$.
Specifically, we can leverage decades of ingenuity
that went into the development of effective automatic abstraction techniques
(such as trace abstraction~\cite{HHP:CAV2013}, interpolation~\cite{interpolation}, CEGAR~\cite{clarke2000counterexample}, 
$k$-Induction, block encoding~\cite{ABE}) for reachability verification tools. 
These tools are unaware of the concept of ``commutativity,'' nor do are they concerned with properties such as memory safety, full-functional correctness, etc.
These are targeted tools, tuned toward proving non-reachability of $q_{\mathit{er}}$.
In this way, when these abstraction strategies are applied to $\autencMmn{\varphi_m^n}$, they 
focus on only what is needed to show non-reachability, and avoid the need to reason
about large aspects of the program that are irrelevant to reachability (and thus irrelevant to commutativity).
%
%

%
\begin{theorem}[Soundness]
\label{thm:soundness}
 For  object implementation $\aut_o$ and resulting encoding $\autencMmn{\varphi_m^n}$, if every run of $\autencMmn{\varphi_m^n}$ avoids $q_{\mathit{er}}$, then $\varphi_m^n$ is a commutativity condition for $m(\Xs)$ and $n(\Ys)$.
\end{theorem}
\begin{proof}[Proof Sketch]
    We prove the theorem by considering a program analysis's proof that $q_{er}$ cannot be reached
    which takes the form of abstractions and invariants. 
    Since $q_{er}$ is not reachable, after reaching $q_3$ it must be the case that the return values agree. This forces $R_\alpha$ to be such that it is an $m,n$-differencing abstraction.
    Since $q_{er}$ is not reachable, after reaching $q_4$, the automaton loops in $q_4$ where $I$ holds followed by non-deterministic call to $m'$ with nondeterministic action of $m'$. Since the observational equivalence relation is the greatest fixpoint, we can conclude that $I$ satisfies the condition for being an abstract observational equivalence relation. 
    Thus, these invariants $R_\alpha$ and $I$ serve as
    the $m,n$-differencing abstraction and abstract observational equivalence relation, respectively.
    Moreover, the encoding also ensures that $\varphi_m^n$ implies $(\alpha,R_\alpha)$ and that
    $(\alpha,R_\alpha)$ implies $(\alpha,I)$.
 Finally, we employ Theorem~\ref{thm:abstraction}.

\end{proof}
%
The above captures the property of the transformation, which is carefully done so that the abstractions satisfy the hypotheses in Theorem \ref{thm:abstraction}. 

\begin{figure}[t]
$\autencMai \equiv$\\
$\left\{\hspace{0.05in}
  \begin{minipage}{5.3in}\footnotesize
    \fbox{\begin{minipage}{2.8in}
      \begin{program}[style=tt,number=true]
      SimpleSet s1 = new SimpleSet(); \label{ln:0} \\
      w\tab hile(*) \{  int t = *; assume (t>0); switch(*) \{ \\
        case 1: s1.add(t);  case 2: s1.isin(t);  \\
        case 3: s1.size();  
        case 4: s1.clear();  \}\} \untab
    \end{program}
    \end{minipage}}\hspace{0.05in}{\bf (A)}\\
    \fbox{\begin{minipage}{2.8in}
      \begin{program}[style=tt,number=true,firstline=5]
    int x = *; int y = *;\\
    assume( $\commuai(\texttt{s1,x,y})$ ); \label{ln:varphi} \\
    SimpleSet s2 = s1.clone(); \label{ln:copy} \hspace{0.5in}\\
    $\rvm$ = s1.$\myM$();    $\rvmn$ = s1.$\myN$();\\
    $\rvn$ = s2.$\myN$();    $\rvnm$ = s2.$\myM$();\\
    assert($\rvm = \rvnm$ \&\& $\rvn=\rvmn$); \label{ln:rv} \\
    $\boxed{\{ R_{\myM}^{\myN}(\texttt{s1},\texttt{s2}) \}}$ \label{ln:q} $\;$
    \end{program}
    \end{minipage}}\hspace{0.05in}{\bf (B)}\\
    \fbox{\begin{minipage}{2.8in}
      \begin{program}[style=tt,number=true,firstline=11]
  wh\tab ile(true) \{ $\boxed{ \{ I_{SS}(\texttt{s1},\texttt{s2}) \} }$ \label{ln:i} \\ 
   in\tab t t = *; assume (t>0); switch(*) \{ \\
  case 1: assert(s1.add(t)\;\;\; == s2.add(t)); \\
  case 2: assert(s1.isin(t) == s2.isin(t)); \\
  case 3: assert(s1.clear()\; == s2.clear()); \\
  case 4: assert(s1.size()\;\;\; == s2.size()); \} \} \label{ln:whiledone} \\ \untab 
      \end{program}
   \end{minipage}}\hspace{0.05in}{\bf (C)}
\end{minipage} \right.$
\caption{\label{fig:encoding} An example of our transformation $\autencMai$, when applied to the source of \addx\ and \isiny. Formally, the result is  an automaton but here it is depicted as a program. When a candidate formula $\commuai$ is supplied to this encoding, a proof of safety entails that
$\commuai$ is a commutativity condition. \mndifferencing\ abstraction $R_{\myM}^{\myN}$ and observational equivalence relation $I_{SS}$ can be discovered automatically using existing tools. }
\end{figure}

\subsection{Example: \ttt{SimpleSet}}

Fig.~\ref{fig:encoding} illustrates the output of applying our encoding to
the methods \addx\ and \isiny\ of the \ttt{SimpleSet} (from Fig.~\ref{fig:simpleset}) denoted $\autencMai$.
It is important to remember that this resulting encoding $\autencMai$ should never be executed. 
Rather, it is designed so that, when a program analysis tool for reachability is applied,
the tool is tricked into performing the reasoning necessary proving commutativity.
We now discuss the key aspects of the encoding, referring to the
corresponding portion of the pseudocode in Fig.~\ref{fig:encoding}.

\noindent
  (A) \emph{Any initial state.} 
    The encoding is designed so that, for any reachable abstract
    state $\sigma$ of the \ttt{SimpleSet} object, there will be a run
    of $\autencMai$ such that the \ttt{SimpleSet} on Line~\ref{ln:copy}
    will be in state $\sigma$. 
    
\noindent
  (B) \emph{Assume $\varphi_m^n$ and consider $m;\!n$ as well as $n;\!m$.}
    A program analysis will consider all runs that eventually
    exit the first loop (we don't care about those that never exit), and the
    corresponding reachable state \ttt{s1}.
    From \ttt{s1}, the encoding assumes that provided
    commutativity condition on Line~\ref{ln:varphi} and then 
    a clone of \ttt{s1} will be materialized.
    Our encoding then will cause a program analysis to consider 
    the effects of the methods applied in each order, and whether or not 
    the return values will match on Line~\ref{ln:rv}.

\noindent
   (C) \emph{Observational equivalence.} Lines~\ref{ln:i}-\ref{ln:whiledone} 
   consider any sequence of
    method calls $m'(\As'),m''(\As''),\ldots$ that could be applied to both \ttt{s1}
    and \ttt{s2}. If observational equivalence does not hold, then
    there will be a run of $\autencMai$ that applies that
    sequence to \ttt{s1} and \ttt{s2}, eventually finding a discrepancy 
    in return values and going to $q_{er}$.

\noindent
The encoding contains two boxed assertions:
$R_{\addx}^{\isiny}$ on Line~\ref{ln:q} must be strong enough
to guard against the possibility of a run reaching $q_{er}$
due to return values.
Similarly,
$I_{SS}$ on Line~\ref{ln:q} must be strong enough
to guard against the possibility of a run reaching $q_{er}$
owing to an observable difference between \ttt{s1} and \ttt{s2}.
Thus, to prove that \ttt{ERROR} is unreachable,
a verification tool needs to have or discover strong
enough such invariants.

\section{Implementation}
\label{sec:impl}

We have developed
\CityProver{},
capable of automatically verifying commutativity conditions of object methods directly from their implementations.
\CityProver{} takes, as input, C source code for the object (represented as a \verb|struct state_t| and series of methods each denoted \verb|o_m(struct state_t this,...)|). This includes the description of the object (the header file), but no specification. Examples of this input can be found \inapx{\ref{apx:benchmarks}}. We have written them as C macros so that our experiments focus on commutativity rather than testing existing tools' inter-procedural reasoning power, which is a separate matter.
Also provided as input to \CityProver{} is a commutativity condition $\varphi_m^n$ and the method names $m$ and $n$. \CityProver{} then implements the encoding $\autencMmn{\varphi_m^n}$, via a program transformation. 
\todo{clone post.}

\begin{figure}
\figboxN{\includegraphics[width=3.1in]{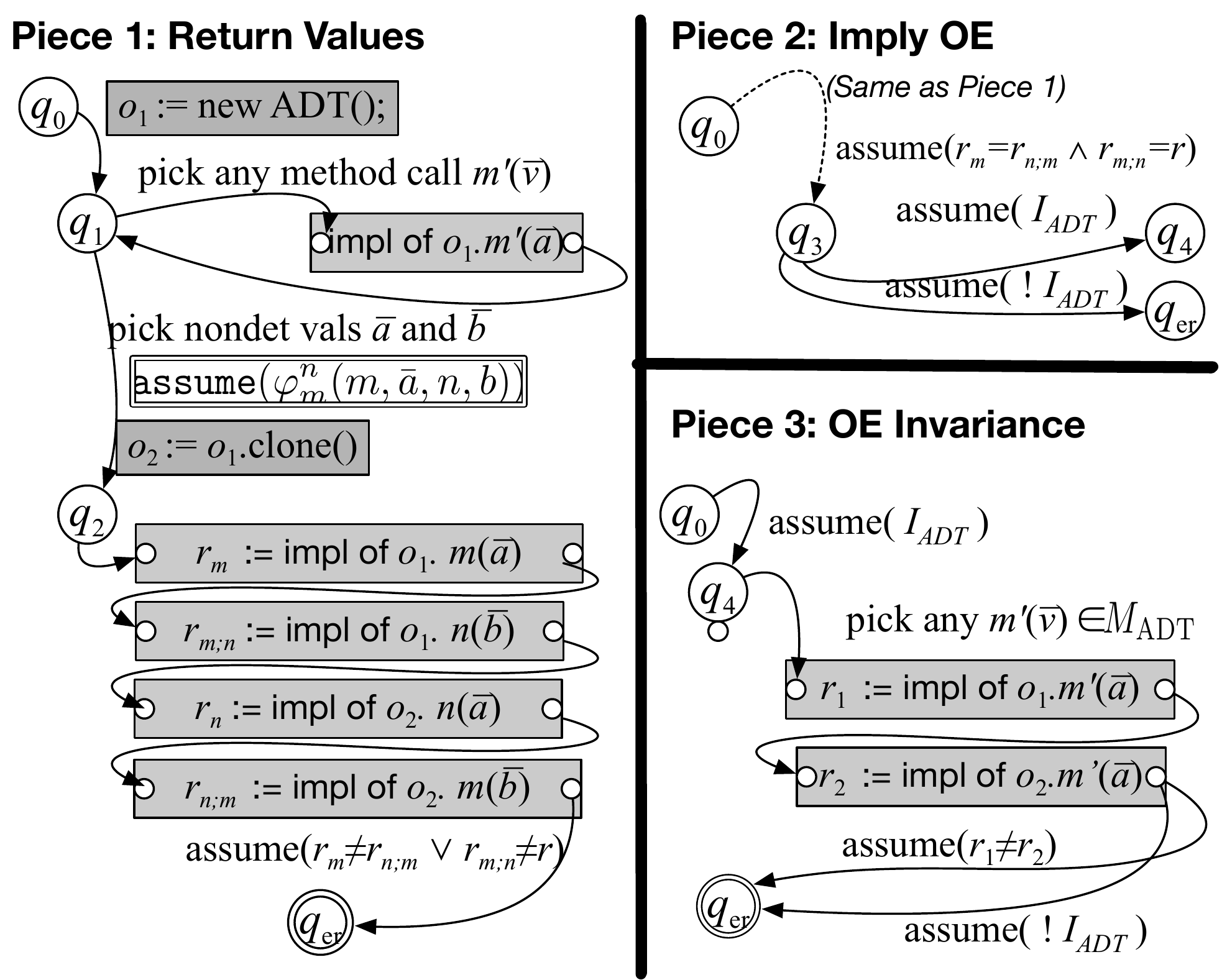}}
\caption{\label{fig:phases} The implementation of our reduction, which generates three separate reachability queries.}
\end{figure}

\paragraph{Abstracting twice.} 
The reduction given in Section~\ref{subsec:transformation} describes an encoding for which a single abstraction can be used. 
However, we found that existing tools struggled to find a single abstraction for reasoning about the multiple aspects of a commutativity proof.
To resolve this issue, we developed an alternative transformation, employed in our implementation, that takes advantage of the \emph{two} abstractions described in Section~\ref{sec:abstraction}:  $\alpha$ for \mndifferencing\ reasoning and $\beta$ for observational equivalence. 

The implementation of our transformation actually generates three different encoding ``pieces'' whose total safety is equivalent to the safety of the single-step reduction in Section~\ref{subsec:transformation}. In this way, different abstractions can be found for each piece.
%
An illustration is given in Fig.~\ref{fig:phases}. 
Piece 1 captures the first part of the general encoding, except the generated
CFA ends just after the instances of $m;n$, $n;m$, with a single \ttt{assume} arc to 
$q_{er}$, with the condition that $r_m \neq r_{nm} \vee r_n \neq r_{mn}$. This piece ensures that $\varphi_m^n$ at least entails that return values agree.
When a program analysis is applied in this piece, it discovers a $(\alpha, R_\alpha)$.
%
Piece 2 repeats some of piece 1's internals but then assumes return value agreement and is used to co\"{e}rce a safety prover to find a sufficient post-condition of $m;n$ and $n;m$ that entails the observational equivalence relation $I_{ADT}$. 
%
%
Finally, Piece 3 is used to co\"{e}rce a safety prover to discover a $\beta$ abstraction in order to prove that $I_{ADT}$ is indeed an observational equivalence relation. This is achieved  by simply \emph{assuming} a pair of states satisfying $I_{ADT}$, and then introducing arcs to $q_{er}$ that will be taken if ever a non-deterministically chosen method could observe different return values or lead to a pair of states that do not satisfy $I_{ADT}$.

Since Piece 3 is analyzed separately from Piece 1, the abstraction $\alpha$ for the \mndifferencing\ relation need not be the same abstraction as $\beta$, which is used for the abstract equivalence relation. Meanwhile, Piece 2 does the work of linking them together: a safety prover would discover a proof that $R_m^n$ implies $I$.



\newcommand\cpac[1]{\textcolor{purple}{#1}}
\newcommand\rTRUE{$\checkmark$}
\newcommand\rUNKNOWN{{\bf ?}}
\newcommand\rTIMEOUT{{\bf T.O.}}
\newcommand\rMEMOUT{{\bf M.O.}}
\newcommand\rIG{n/a}
\newcommand\rFALSE{$\chi$}
\newcommand\rgood{ok}
\newcommand\rFAIL{} 
\newcommand\ttM[1]{\texttt{#1}}
\newcommand\MM[2]{\ttM{#1}\bowtie\ttM{#2}}
\renewcommand\oe{}
\newcommand\rv{}
\newcommand\oeA{i}
\newcommand\oeB{ii}

%

\begin{figure*}[t]
\centering
  \setlength{\tabcolsep}{3pt} 
\scalebox{0.90}{
  \begin{tabular}{|l|c|p{0.75in}|c|rc|r|c|rc|r|c|}
    \hline
        {\bf ADT} & {\bf Meths}.~$m(x_1),n(y_1)$ &  $\varphi_{m(x_1)}^{n(y_1)}$  & {\bf Exp.} & \multicolumn{2}{c|}{{\bf CPAchecker} (by Piece)}  & $\Sigma$ & {\bf Out} & \multicolumn{2}{c|}{{\bf Ultimate} (by Piece)}  & $\Sigma$ & {\bf Out} \\
    \hline
    \input{bothfinal}
    \hline
  \end{tabular}}
  \caption{\label{fig:easy} Results of applying \CityProver{} to the small benchmarks. We report results when using CPAchecker as well as when using Ultimate.}
\end{figure*}

\newcommand\phiasA{{}^1\!\varphi_\texttt{push}^\texttt{pop}}

\newcommand\phihti{{}^i\!\varphi_\texttt{put}^\texttt{put}}
\newcommand\phihtA{{}^1\!\varphi_\texttt{put}^\texttt{put}}
\newcommand\phihtB{{}^2\!\varphi_\texttt{put}^\texttt{put}}
\newcommand\phihtC{{}^3\!\varphi_\texttt{put}^\texttt{put}}

\begin{figure*}[t]\centering
\scalebox{0.9}{\centering
  \begin{tabular}{|l|c|p{1.9in}|c||rc|r|c|}
\hline
        {\bf ADT} & {\bf Meths}.~$m(x_1),n(y_1)$ &  $\varphi_{m(x_1)}^{n(y_1)}$ & {\bf Exp.} & \multicolumn{2}{c|}{{\bf Ultimate} (by Piece)}  & $\Sigma$ & {\bf Out} \\
    \hline
    \input{hardult2}
    \hline
  \end{tabular}}
  \begin{center}
  $\begin{array}{ll}
  \phiasA \equiv & o_1.a[o_1.top] = x_1 \wedge o_1.top > 1 \wedge o_1.top < MAX \qquad\qquad\qquad\qquad\qquad\phihtA \equiv x_1 \neq y_1 \\
\phihtB \equiv & x_1 \neq y_1 \wedge o_1.table[x_1\%CAP].key = -1 \wedge o_1.table[y_1\%CAP].key = -1\\
\phihtC \equiv & x_1 \neq y_1 \wedge x_1\%CAP \neq y_1\%CAP \wedge o_1.table[x_1\%CAP].key = -1 \wedge o_1.table[y_1\%CAP].key = -1\\
\end{array}$\end{center}
  \caption{\label{fig:hard} Results of applying \CityProver{} to the more involved data structures, using Ultimate as a reachability solver. Longer commutativity conditions are defined below the table. Note that $o_1.x$ is the object field, and $x_1,y_1$ are $m,n$ parameters, resp.}
\end{figure*}



\section{Experiments}
\label{sec:eval}

Our first goal was to evaluate whether, through our transformation, we can use existing tools to verify commutativity conditions. To this end, we created a series of  benchmarks (with integers, simple pointers, structs and arrays) and ran \CityProver{} on them.
Our experiments were run on a Quad-Core
Intel(R) Xeon(R) CPU E3-1220 v6 at 3.00 GHz, inside a QEMU virtual host, running Ubuntu.

\paragraph{Simple benchmarks.}
We began with simple ADTs including:
a \textsf{Memory} that has a single memory cell, with methods \texttt{read(x)} and \texttt{write(x,v)};
an \textsf{Accumulator} with increment, decrement, and a check whether the value is 0;
and a \textsf{Counter}, like Accumulator, but that also has a \ttt{clear} method.
For each object, we considered some example method pairs with a valid commutativity condition (to check that the tool proves it to be so) and an incorrect commutativity condition (to check that the tool discovers a counterexample).
The objects, method pairs, and commutativity conditions are shown in Fig.~\ref{fig:easy}.

We ran \CityProver{} first using CPAchecker~\cite{cpachecker} and then using Ultimate~\cite{ultimate}, reporting the results in the subsequent columns of Fig.~\ref{fig:easy}. 
For CPAchecker we used the Z3~\cite{z3} SMT solver and (a minor adaptation of) the \texttt{predicateAnalysis} configuration.
For each tool, we report the time taken to prove/disprove each of the three Pieces (Section~\ref{sec:impl}), as well as the result for each piece ($\checkmark$ means the encoding was safe, and $\chi$ means a counterexample was produced). Finally, we report the final output of \CityProver{} and the overall time in seconds.
These experiments confirm that both CPAchecker and Ultimate can prove commutativity of these simple benchmarks, within a few seconds. In one case, CPAchecker returned an incorrect result.

\paragraph{Larger benchmarks.}
We next turned to data structures that store and manipulate elements, and considered the ADTs listed below with various candidate commutativity conditions.
For these ADTs, we had difficulty tuning CPAchecker (perhaps owing to our limited experience with CPAchecker) so we only report our experience using Ultimate as a reachability solver. The results of these benchmarks are given in Fig.~\ref{fig:hard}. We discuss each ADT in turn.

\begin{itemize*}

\item \textsf{ArrayStack}. (Fig.~\ref{fig:simpleset}) Note that push/pop commutativity condition $\phiasA$ is defined below the table in Fig.~\ref{fig:hard}. For this ADT, some method pairs and commutativity conditions took a little longer than others.

\item \textsf{SimpleSet}. (Fig.~\ref{fig:simpleset}) Most cases where straightforward; not surprisingly, \addx{}/\isiny{} took more time.

\item A simple \textsf{HashTable}, where hashing was done only once and insertion gives up if there is a collision. Here again, note the definitions of some commutativity conditions below the table.

\item \textsf{Queue}, implemented with an array. \CityProver{} was able to prove all but two commutativity conditions for the Queue. 
The \ttt{enq}/\ttt{deq} case ran out of memory.
\todo{fix}

\item \textsf{List Set}.  Some commutativity conditions could be quickly proved or refuted. For more challenging examples, \eg\ \ttt{add}/\ttt{rm}, Ultimate ran out of memory. 
In these cases reachability involves reasoning between lists with potentially different element orders. We believe that the necessary  disjunctive reasoning overwhelmed Ultimate. 


\end{itemize*}

In summary, \CityProver{} allowed us to prove a variety of commutativity conditions using unmodified Ultimate. The power of \CityProver{} greatly depends on the power of said underlying verification tool. 

For this prototype, we provided observational equivalence relations manually.
As tools improve their ability to synthesize loop invariants, this input may become unnecessary.

\subsection{Usability}
\label{subsec:usability}

We now discuss the usability of \CityProver{}, including
lessons learnt from our experience.

\emph{Does the approach avoid the need for full specifications?}
Yes. \mndifferencing\ focuses on the aspects of data structure implementations that pertain to commutativity and need not  capture full functional correctness. Since \mndifferencing\ is based on method orderings from \emph{the same} start state, many details of a specification become trivial. Notice that the commutativity conditions in Fig.~\ref{fig:hard} are far simpler than what the full specifications of those data structures would be. 
The ability to proceed without specifications is also convenient because, for many data structures that occur in the wild, specifications are not readily available.

\emph{Are commutativity formula easier to come up with than specifications?}
Intuitively, yes, because  commutativity conditions can often be very compact (\eg\ $x\neq y$). 
By contrast, specifications relate every post-state with its corresponding pre-state and typically involve multiple lines of logical formulae, not easily written by users. 
Since commutativity conditions are smaller, users can make a few guesses until they find one that holds.

We followed this methodology in some of the benchmarks:
(i)  \ttt{SimpleSet}. For $\ttt{isin}(x)\bowtie\ttt{clear}$, we originally thought that 
$\ttt{a} \neq x_1 \wedge \ttt{b} \neq y_1$ would be a valid commutativity condition, expressing that neither \ttt{a} nor \ttt{b} had the value. The tool caught us, and reported a counterexample, showing the typo ($y_1$ should be $x_1$).
(ii) \ttt{Queue}. For \ttt{enq}/\ttt{enq}, we guessed ``true'' as a first condition. \CityProver{} then produced a counterexample, showing they do not commute if the arguments are different. We strengthened the condition, but that still was not enough: the next counterexample showed they do not commute if there's not enough space.
(iii) \ttt{HashTable}. The successive conditions $\phihti$ represent our repeated attempts to guess commutativity conditions for $\ttt{put}(x_1)\bowtie\ttt{put}(y_1)$.
We first tried a simple condition $\phihtA$ expressing that the keys are distinct. \CityProver{} returned a counterexample, showing that commutativity doesn't hold because capacity may be reached. Next, we created $\phihtB$ expressing that the destination slots for $x_1$ and $y_1$ were both occupied. Again, \CityProver{} reported a counterexample that arises when the keys \emph{collide}. Finally, $\phihtC$ requires distinct buckets. 


\section{Related work}
\label{sec:relwork}

Bansal \emph{et al.}~\cite{tacas18} describe how to synthesize commutativity conditions from provided pre/post specifications, rather than from implementations. They assume that these specifications are precise enough to faithfully represent all effects that are relevant to the ADT's commutativity.
By contrast, our work does not require the data structure's specification
 to be provided which is convenient because specifications are typically not readily available for \emph{ad hoc} data
 structures. 
 
Moreover, it is not straight forward to infer the specification (\eg\ via static analysis) because the appropriate precision is not known \emph{apriori}. 
On one hand, a sound but imprecise specification can lead to incorrect commutativity conditions:
for a Stack, specification \texttt{\{true\}push(v)\{true\}} is sound but does not capture effects of \texttt{push} that
 are relevant to commutativity with, say, \texttt{pop}. With such an imprecise specification, Bansal \emph{et al.} would emit incorrect commutativity conditions.
On the other hand, 
a sound but overly precise specification can be burdensome and unnecessary for
 commutativity reasoning:
specification \texttt{\{$A,n$\}push(v)\{$A'. A[n+1]=v \wedge \forall i\leq n. A'[i] = A[i], n+1$\}} captures too much information about the stack (deep values in the stack) that is irrelevant to commutativity.
By contrast, in this paper we capture just what is needed
for commutativity, \eg, only the top and second-to-top elements of the stack.

Gehr \emph{et al.}~\cite{DBLP:conf/cav/GehrDV15} describe a method based on black-box sampling. They draw concrete arguments, extract relevant predicates from the sampled set of examples, and then search for a formula over the predicates. There  is no soundness or completeness guarantee.
Both Aleen and Clark~\cite{aleen} and Tripp \emph{et al.}~\cite{Tripp2012} identify sequences of actions that commute (via random interpretation and dynamic analysis, respectively). However, neither technique yields an explicit commutativity condition.
Kulkarni et al. \cite{KNPSP:PLDI11} point out that 
varying degrees of commutativity specification precision are useful.
Kim and Rinard \cite{KR:PLDI11} use Jahob to
verify manually specified commutativity conditions 
of several different linked data structures.
Commutativity specifications are also
found in dynamic analysis techniques~\cite{DBLP:conf/pldi/DimitrovRVK14}.
%
Najafzadeh \emph{et al.}~\cite{NGYFS2016} describe a tool for weak consistency, which
reports some commutativity checking of formulae. It does not appear to apply to 
data structure implementations.

Reductions to reachability are common in verification. They have been
used in security---so called self-composition~\cite{Barthe2004,Terauchi2005}---reduces (some forms of)
hyper-properties~\cite{Clarkson2010} to properties of a single program.
Others have reduced
crash recoverability~\cite{Koskinen2016},
temporal verification~\cite{Cook2011},
and linearizability~\cite{Bouajjani2015}
to reachability.

 

\section{Conclusion}
\label{sec:conclusion}

We have described a theory, algorithm and tool that, together, allow one to automatically verify commutativity conditions of data structure implementations. We describe an abstraction called \mndifferencing, that focuses the verification question on the differences in the effects of $m$ and $n$ that depend on the order in which they are applied, abstracting away effects that would have been the same regardless of the order. We further showed that verification of commutativity conditions can be reduced to reachability. Finally, we described \CityProver{}, the first tool capable of verifying commutativity conditions of data structure implementations.

The results of our work can  be used to  incorporate more commutativity conditions soundly and obtain speed ups in transactional object systems~\cite{aplas19,ppopp08}.
Further research is needed to use our commutativity proofs with parallelizing compilers.
Specifically, in the years to come, parallelizing compilers could combine our proofs of commutativity with automated proofs of linearizability~\cite{Bouajjani2015} to  execute more code concurrently and safely.


\bibliography{biblio}

\vfill
\pagebreak
\appendix
\onecolumn
\section{Benchmark Sources}
\label{apx:benchmarks}
\lstset{language=C}
\subsection{Memory}
\begin{lstlisting}
struct state_t { int x; };
\end{lstlisting}
\begin{lstlisting}
#include <stdlib.h>

#define o_new(res) res = (struct state_t *)malloc(sizeof(struct state_t))

#define o_clone(src) { \
  struct state_t* tmp; \
  o_new(tmp); \
  tmp->x = (src)->x;                            \
  dst = tmp; \
  }

#define o_read(st,rv)    rv = (st)->x
#define o_write(st,rv,v) { (st)->x = v; rv = 0; }


\end{lstlisting}
\subsection{Counter}
\begin{lstlisting}
struct state_t { int x; };
\end{lstlisting}
\begin{lstlisting}
#include <stdlib.h>

#define o_new(res) res = (struct state_t *)malloc(sizeof(struct state_t))
#define o_close(src) {                          \
    struct state_t* tmp;                        \
    o_new(tmp);                                 \
    tmp->x = src->x;                            \
    dst = tmp;                                  \
  }
#define o_incr(st,rv) { (st)->x += 1; rv = 0; }
#define o_clear(st,rv) { (st)->x = 0; rv = 0; }
#define o_decr(st,rv) { if ((st)->x == 0) rv = -1; else { (st)->x -= 1; rv = 0; } }
#define o_isz(st,rv) { rv = ((st)->x == 0 ? 1 : 0); }
\end{lstlisting}
\subsection{Accumulator}
\begin{lstlisting}
struct state_t { int x; };
\end{lstlisting}
\begin{lstlisting}
#include <stdlib.h>

#define o_new(res) res = (struct state_t *)malloc(sizeof(struct state_t))
#define o_close(src) {                          \
    struct state_t* tmp;                        \
    o_new(tmp);                                 \
    tmp->x = src->x;                            \
    dst = tmp;                                  \
  }
#define o_incr(st,rv) { (st)->x += 1; rv = 0; }
#define o_decr(st,rv) { (st)->x -= 1; rv = 0; }
#define o_isz(st,rv) { rv = ((st)->x == 0 ? 1 : 0); }
\end{lstlisting}
\subsection{Array Queue}
\begin{lstlisting}
#define MAXQUEUE 5

struct state_t { int a[MAXQUEUE]; int front; int rear; int size; };
\end{lstlisting}
\begin{lstlisting}
#include <stdlib.h>

#define o_new(res) { \
  res = malloc(sizeof(struct state_t)); \
  res->front = 0; \
  res->rear = MAXQUEUE-1; \
  res->size = 0; \
}

#define o_enq(st,rv,v) { \
  if((st)->size == MAXQUEUE) rv = 0; \
  else { \
    (st)->size++; (st)->rear = ((st)->rear + 1) % MAXQUEUE; (st)->a[(st)->rear] = v; rv = 1; } }

#define o_deq(st,rv) { \
  if((st)->size==0) rv = -1; \
  else { int r = (st)->a[(st)->front]; \
         (st)->front = ((st)->front + 1) % MAXQUEUE; \
         (st)->size--; \
         rv = r; } }

#define o_isempty(st,rv) rv = ( (st)->size==0 ? 1 : 0)
\end{lstlisting}
\subsection{Array Stack}
\begin{lstlisting}
#define MAXSTACK 5

struct state_t { int a[MAXSTACK]; int top; };
\end{lstlisting}
\begin{lstlisting}
#define o_new(res) res = (struct state_t *)malloc(sizeof(struct state_t)); res->top = -1;

#define o_push(st,rv,v) { \
  if ((st)->top == (MAXSTACK-1)) {rv = 0;} \
  else { (st)->top++; (st)->a[(st)->top] = v; rv = 1; } }

#define o_pop(st,rv) {          \
  if ((st)->top == -1) rv = -1; \
  else rv = (st)->a[ (st)->top-- ]; }

#define o_isempty(st,rv) rv = ((st)->top==-1 ? 1 : 0)
\end{lstlisting}
\subsection{Hash Table}
\begin{lstlisting}
#define HTCAPACITY 11
struct entry_t { int key; int value; };
struct state_t { struct entry_t table[HTCAPACITY]; int keys; };
\end{lstlisting}
\begin{lstlisting}
#include <stdlib.h>

#define o_new(st) { s1 = malloc(sizeof(struct state_t)); \
    for (int i=0;i<HTCAPACITY;i++) { s1->table[i].key = -1; } \
    s1->keys = 0; }

#define o_put(st,rv,k,v) { int slot = k % HTCAPACITY; \
    if((st)->table[slot].key == -1) {		      \
      (st)->table[slot].key = k;			      \
      (st)->table[slot].value = v;			      \
      (st)->keys++;				      \
      rv= 1; }					      \
    else if((st)->table[slot].key == k) {		      \
      (st)->table[slot].value = v;			      \
      rv = 1; }					      \
    else if ((st)->table[slot].key != k) {		      \
      rv = -1; }				      \
    else { rv = -1; } }

#define o_get(st,rv,k) { int slot = k % HTCAPACITY; \
    if((st)->table[slot].key != k) rv = -1;	    \
    else rv = (st)->table[slot].value; }

#define o_rm(st,rv,k) { int slot = k % HTCAPACITY;	\
  if((st)->table[slot].key == k) {				\
    (st)->table[slot].key = -1; rv = 1; }			\
  else { rv = -1; } }

#define o_isempty(st,rv) { if((st)->keys==0) rv = 1; else rv = 0; }

\end{lstlisting}
\subsection{Simple Set}
\begin{lstlisting}

struct state_t { int a; int b; int sz; };
\end{lstlisting}
\begin{lstlisting}
#include <assert.h>
#include <stdlib.h>

#define o_new(st) {  \
  st = malloc(sizeof(struct state_t)); \
  st->a = -1; st->b = -1; st->sz = 0; }

#define o_add(st,rv,v) { \
  if ((st)->a == -1 && (st)->b == -1) { (st)->a = v; (st)->b=(st)->b+0;  (st)->sz++; rv = 0; } \
  else if ((st)->a != -1 && (st)->b == -1) { (st)->b = v; (st)->a=(st)->a+0; (st)->sz++; rv = 0; } \
  else if ((st)->a == -1 && (st)->b != -1) { (st)->a = v; (st)->b=(st)->b+0; (st)->sz++; rv = 0; } \
  else { rv = 0; } }

#define o_isin(st,rv,v) { \
  rv = 0; \
  if ((st)->a == v) rv = 1; \
  if ((st)->b == v) rv = 1; }

#define o_getsize(st,rv) { rv = (st)->sz; }

#define o_clear(st,rv) { (st)->a = -1; (st)->b = -1; (st)->sz = 0; rv = 0; }

#define o_norm(st,rv) { \
  if ((st)->a > (st)->b) { \
    int t = (st)->b; \
    (st)->b = (st)->a; \
    (st)->a = t; \
  } \
  rv = 0; }
\end{lstlisting}
\subsection{List}
\begin{lstlisting}
#define LISTCAP 5
struct state_t { int list[LISTCAP]; int end; };

/*struct state_t* o_clone(struct state_t* s1);
int o_add(struct state_t* s, int v);
int o_rm(struct state_t* s, int v);
int o_contains(struct state_t* s, int v);
int o_isempty(struct state_t* s);
int o_print(struct state_t* s);*/
\end{lstlisting}
\begin{lstlisting}
#include <stdio.h>
#include <stdlib.h>

#define o_new(st) { st = malloc(sizeof(struct state_t)); st->end=-1; }

#define o_add(st,rv,v) {                        \
  rv = 1; \
  for (int i=0; i<=(st)->end; i++) {            \
    if ((st)->list[i] == v) { rv = 0; break; } \
  } \
  if (rv == 1 && (st)->end == LISTCAP) { rv = 0; } \
  else { (st)->list[(st)->end+1] = v;  (st)->end++; } }


#define o_rm(st,rv,v) {                         \
  rv = 0;                                       \
  for (int i =0; i<=(st)->end; i++) { \
      if ( (st)->list[i] == v) { \
        (st)->list[i] = (st)->list[(st)->end]; (st)->end--; rv = 1; break; \
      } } }

#define o_contains(st,rv,v) {                                                 \
    rv = 0;                                                             \
    for (int i =0; i<=(st)->end; i++) {                                 \
      if ( (st)->list[i] == v) {  rv = 1; break; } \
    } }

#define o_isempty(st,rv) rv = ((st)->end > -1 ? 1 : 0)
\end{lstlisting}

\end{document}